\documentclass[a4paper,USenglish,autoref]{lipics-v2021}

\usepackage{algorithm}
\usepackage[noend]{algpseudocode}
\usepackage{booktabs}

\newcommand{\bigo}[1]{\ensuremath{\mathcal{O}(#1)}}
\renewcommand{\Pr}[1]{\ensuremath{\text{Pr}\left[#1\right]}}
\newcommand{\true}{\ensuremath{\textsc{true}}}
\newcommand{\false}{\ensuremath{\textsc{false}}}
\newcommand{\alg}{\ensuremath{\mathcal{A}}}
\newcommand{\lbl}{\ensuremath{\ell}}
\newcommand{\disconnectset}{\ensuremath{D}}
\newcommand{\lock}{\ensuremath{\texttt{lock}}}
\newcommand{\state}{\ensuremath{\texttt{state}}}
\newcommand{\phase}{\ensuremath{\texttt{phase}}}
\newcommand{\lockset}{\ensuremath{L}}
\newcommand{\replyset}{\ensuremath{R}}
\newcommand{\winset}{\ensuremath{W}}
\newcommand{\holdset}{\ensuremath{H}}
\newcommand{\applyset}{\ensuremath{A}}
\newcommand{\candidateset}{\ensuremath{C}}
\newcommand{\compete}{\ensuremath{P}}
\newcommand{\msg}[2]{\ensuremath{\texttt{#1(#2)}}}
\newcommand{\Lock}{\ensuremath{\textsc{Lock}}}
\newcommand{\Unlock}{\ensuremath{\textsc{Unlock}}}

\makeatletter
\algrenewcommand\ALG@beginalgorithmic{\small}
\algrenewcommand\alglinenumber[1]{\scriptsize #1:}
\makeatother

\algdef{SE}[SUBALG]{Indent}{EndIndent}{}{\algorithmicend\ }%
\algtext*{Indent}
\algtext*{EndIndent}
\makeatletter
\newcommand{\multiline}[1]{%
  \begin{tabularx}{\dimexpr\linewidth-\ALG@thistlm}[t]{@{}X@{}}
    #1
  \end{tabularx}
}
\makeatother

\newif\ifarxiv
\arxivtrue

\newif\ifcomment
\commenttrue    

\newif\iffigabbrv
\figabbrvfalse   
\newcommand{\figtext}{\iffigabbrv Fig.\else Figure\fi}

\title{Local Mutual Exclusion for Dynamic, Anonymous, Bounded Memory Message Passing Systems}
\titlerunning{Local Mutual Exclusion for Dynamic, Limited Message Passing Systems}

\author{Joshua J. Daymude}{Biodesign Center for Biocomputing, Security and Society, Arizona State University, Tempe, AZ}{jdaymude@asu.edu}{https://orcid.org/0000-0001-7294-5626}{NSF (CCF-1733680), U.S.\ ARO (MURI W911NF-19-1-0233), the Momental Foundation's Mistletoe Research Fellowship, and the ASU Biodesign Institute.}
\author{Andr\'ea W. Richa}{School of Computing and Augmented Intelligence, Arizona State University, Tempe, AZ}{aricha@asu.edu}{https://orcid.org/0000-0003-3592-3756}{NSF (CCF-1733680, CCF-2106917) and U.S.\ ARO (MURI W911NF-19-1-0233).}
\author{Christian Scheideler}{Department of Computer Science, Paderborn University, Paderborn, Germany}{scheideler@upb.de}{https://orcid.org/0000-0002-5278-528X}{DFG Project SCHE 1592/6-1.}
\authorrunning{J.\ J.\ Daymude, A.\ W.\ Richa, and C.\ Scheideler}

\Copyright{Joshua J.\ Daymude, Andr\'ea W.\ Richa, and Christian Scheideler}



\ifarxiv
\hideLIPIcs
\else\fi
\nolinenumbers

\ifarxiv\else
\EventEditors{James Aspnes and Othon Michail}
\EventNoEds{2}
\EventLongTitle{1st Symposium on Algorithmic Foundations of Dynamic Networks (SAND 2022)}
\EventShortTitle{SAND 2022}
\EventAcronym{SAND}
\EventYear{2022}
\EventDate{March 28--30, 2022}
\EventLocation{Virtual Conference}
\EventLogo{}
\SeriesVolume{221}
\ArticleNo{12}
\fi

\begin{document}

\maketitle

\begin{abstract}
    \textit{Mutual exclusion} is a classical problem in distributed computing that provides \textit{isolation} among concurrent action executions that may require access to the same shared resources.
    Inspired by algorithmic research on distributed systems of weakly capable entities whose connections change over time, we address the \textit{local mutual exclusion} problem that tasks each node with acquiring exclusive locks for itself and the maximal subset of its ``persistent'' neighbors that remain connected to it over the time interval of the lock request.
    Using the established \textit{time-varying graphs} model to capture adversarial topological changes, we propose and rigorously analyze a local mutual exclusion algorithm for nodes that are anonymous and communicate via asynchronous message passing.
    The algorithm satisfies \textit{mutual exclusion} (non-intersecting lock sets) and \textit{lockout freedom} (eventual success with probability 1) under both semi-synchronous and asynchronous concurrency.
    It requires $\bigo{\Delta}$ memory per node and messages of size $\Theta(1)$, where $\Delta$ is the maximum number of connections per node.
    We conclude by describing how our algorithm can implement the pairwise interactions assumed by \textit{population protocols} and the concurrency control operations assumed by the \textit{canonical amoebot model}, demonstrating its utility in both passively and actively dynamic distributed systems.
\end{abstract}

\section{Introduction} \label{sec:intro}

Distributed computing research has grown increasingly concerned with characterizing the capabilities and limitations of systems composed of \textit{dynamic entities} (or \textit{nodes}).
Recently, these studies have considered both biological collectives such as social insects~\cite{AndresArroyo2018-stochasticapproach,Chandrasekhar2018-trailrepair,Ghaffari2015-distributedhousehunting}, spiking neural networks~\cite{Su2019-spikebasedwinnertakeall}, and DNA and molecular computers~\cite{Chalk2018-freezingsimulates,Patitz2014-introductiontilebased,Woods2013-activeselfassembly} as well as engineered systems such as overlay networks and the Internet of Things (IoT)~\cite{Feldmann2020-surveyalgorithms}, swarm and modular self-reconfigurable robotics~\cite{Flocchini2019-distributedcomputing,Hamann2018-swarmrobotics,Piranda2018-designingquasispherical,Yim2007-modularselfreconfigurable}, and programmable matter~\cite{Angluin2006-computationnetworks,Daymude2021-canonicalamoebot,Derakhshandeh2014-amoebotba,Michail2016-simpleefficient}.
Entities in these systems often make decisions based only on their own knowledge (or ``state''), locally-perceptible measures of their environment (e.g., pheromones, the number or density of nearby neighbors, etc.), and information communicated to them by their neighbors.

Compared to the static setting where acting nodes' neighborhoods do not change, designing correct distributed algorithms in the dynamic setting is a challenging task.
In this paper, we use the established \textit{time-varying graphs} (TVGs) model~\cite{Casteigts2018-journeythrough,Casteigts2012-timevaryinggraphs} to capture adversarial changes in network topology and consider weakly capable nodes that are \textit{anonymous}, have \textit{bounded memory}, communicate via \textit{asynchronous message passing}, and execute their algorithms \textit{semi-synchronously} or \textit{asynchronously}.
The classical \textit{mutual exclusion} problem~\cite{Dijkstra1965-mutualexclusion} regulates how nodes enter their critical sections using locks, defined as a pair of operations \Lock\ and \Unlock.
Our \textit{local mutual exclusion} problem---designed to enable nodes to locally coordinate their interactions in the dynamic, concurrent setting---defines \Lock\ as a node acquiring locks for itself and the maximal subset of its ``persistent'' neighbors that remain connected to it while the request is processed.
A core challenge in designing such a \Lock\ operation in the dynamic setting lies in the nodes' inability to know, when issuing lock requests, which neighbors will be persistent and which others will later be removed.

This locking mechanism greatly simplifies the design of local distributed algorithms in highly dynamic settings by providing \textit{isolation} for concurrently executed actions.
An algorithm's actions can first be designed for the simpler sequential setting in which at most one node is active (potentially changing the system configuration) at a time.
When considering the concurrent setting, each action is then treated as a critical section wrapped in a \Lock/\Unlock\ pair; this ensures that no two simultaneously executing actions can involve overlapping neighborhoods.
Our locking mechanism gracefully handles neighbor disconnections, ensuring that the locked and connected subset of an acting node's neighborhood remains fixed throughout the execution of its action, just as it would be in the sequential setting.
Thus, our locking mechanism restricts the algorithm designer's concern from all possible complications arising from concurrent dynamics to just one: New connections may concurrently be established with a node while it is executing an action.

\subparagraph*{Our Contributions.}
We summarize our contributions as follows.
\begin{itemize}
    \item A formalization of the \textit{local mutual exclusion problem} in an extension of the time-varying graphs model that captures topological changes, asynchronous message passing, and semi-synchronous or asynchronous node activation (Section~\ref{sec:prelims}).
    
    \item A randomized algorithm implementing the \Lock\ and \Unlock\ operations for local mutual exclusion that satisfies \textit{mutual exclusion} (non-intersecting lock sets) and \textit{lockout freedom} (eventual success with probability 1) under both semi-synchronous and asynchronous concurrency.
    This algorithm requires $\bigo{\Delta}$ memory per node and messages of size $\Theta(1)$, where $\Delta$ is the maximum number of connections per node (Sections~\ref{sec:alg}--\ref{sec:async}).
    
    \item Applications of this local mutual exclusion algorithm to \textit{population protocols}~\cite{Angluin2006-computationnetworks}, establishing an underlying mechanism for guaranteeing pairwise interactions in a broader class of concurrent activation models, and the \textit{canonical amoebot model}~\cite{Daymude2021-canonicalamoebot}, implementing the model's concurrency control operations (Section~\ref{sec:applications}).
\end{itemize}

\paragraph*{Related Work}

Designing algorithms for concurrent computing environments is a challenging task requiring the careful control of simultaneously interacting processes and coordinated access to shared resources.
Since its introduction by Dijsktra~\cite{Dijkstra1965-mutualexclusion}, the closely related \textit{mutual exclusion} problem has received much attention from the research community.
For shared memory systems, mutual exclusion can be conveniently solved by atomic operations like compare-and-swap, test-and-set, and fetch-and-add~\cite{Herlihy1991-waitfreesynchronization}.
In contrast, our present focus is on asynchronous message passing.
Classical approaches to mutual exclusion in asynchronous message passing systems often assume that nodes have unique identifiers and global coordination (see, e.g., the survey~\cite{Singhal1993-taxonomydistributed}) or make use of unbounded counters like Lamport clocks (e.g.,~\cite{Maekawa1985-sqrtnmutex,Ricart1981-optimalalgorithm}), neither of which are appropriate for the anonymous, bounded memory nodes we consider here.
The most relevant classical algorithm to our setting is the \textit{arrow protocol}~\cite{Demmer1998-arrowdistributed,Raymond1989-treemutex} that requires only constant memory per node to locally maintain a spanning tree rooted at the node with exclusive access to the shared resource; however, despite recent improvements~\cite{Ghodselahi2017-dynamicanalysis,Khanchandani2019-arvydistributed}, it is not clear how to adapt this protocol to systems with dynamic topologies.

Our local variant of the mutual exclusion problem blurs the usual delineation between processes and the shared resources they're accessing as nodes compete to gain exclusive access to their neighborhoods.
Like the well-studied \textit{$k$-mutual exclusion}~\cite{Fischer1979-resourceallocation} and \textit{group mutual exclusion}~\cite{Joung2000-asynchronousgroup,Joung2002-congenialtalking} variants, ours allows multiple nodes to be in their critical sections simultaneously; however, these variants allow multiple process to access the same shared resource(s) concurrently while ours requires that concurrently locked neighborhoods be non-intersecting.
This constraint is similar to ensuring the active nodes form a \textit{distance-3 independent set} from graph theory and is related to the more general \textit{$(\alpha, \beta)$-ruling sets}~\cite{Awerbuch1989-networkdecomposition} recently solved under the \textsf{LOCAL} and \textsf{CONGEST} models~\cite{Kuhn2018-deterministicdistributed,Schneider2013-symmetrybreaking}; however, these distributed algorithms rely on static topologies, unique identifiers, and synchronous message delivery.
The recent results on mutual exclusion for \textit{fully anonymous} systems~\cite{Raynal2020-mutualexclusion}, like our nodes and their neighborhoods, assume that neither the processes nor the shared resources have unique identifiers.
However, like the earlier classical results above and other recent models of weak finite automata~\cite{Emek2013-stoneage,Esparza2020-classificationweak}, these do not extend to dynamic network topologies.

Research on \textit{mobile ad hoc networks} (MANETs) directly embraces node and edge dynamics, modeling wireless communication links that form and fail as nodes move in and out of each other's transmission radii.
Mutual exclusion has been exhaustively studied under MANET models~\cite{Attiya2010-efficientrobust,Baldoni2002-distributedmutual,Benchaiba2004-distributedmutual,Chen2005-selfstabilizingdynamic,Sharma2014-tokenbased,Walter2001-mutualexclusion}, and many of those ideas inspired recent work on mutual exclusion for intersection traffic control for autonomous vehicles~\cite{Shehu2020-distributedmutual,Wu2015-distributedmutual}.
Mutual exclusion for MANETs is almost always solved using a token-based approach, sometimes combined with the imposition of a logical structure like a ring or tree.
These approaches only apply to competitions for a single shared resource or critical section per token type; our nodes' competitions over their local neighborhoods would need one token type per neighboring node which is not addressed by prior work.
More relevant to our local variant of mutual exclusion are randomized backoff mechanisms for local contention resolution used by MANETs and wireless networks~\cite{Bender2005-adversarialcontention,Cali2000-ieee80211,Capetanakis1979-treealgorithms} to ensure no two nodes are broadcasting in overlapping neighborhoods; however, these rely on nodes' chosen backoff delays to correspond to a consistent wall clock that is incompatible with our weaker model of concurrency.
In any case, the standard MANET communication model of wireless broadcast with time-ordered, instantaneous receipt of messages is more powerful than our asynchronous message passing.
Like MANETs, algorithms for \textit{self-stabilizing overlay networks} (see~\cite{Feldmann2020-surveyalgorithms} for a recent survey) similarly embrace node and edge dynamics, but often use more memory than our present algorithm and assume unique node identifiers.

Finally, we briefly highlight related models of \textit{dynamic networks}, i.e., those whose structural properties change over time.
Our model is closely related to the \textit{time-varying graphs} (TVGs) model~\cite{Casteigts2018-journeythrough,Casteigts2012-timevaryinggraphs} which unifies prior models of dynamic networks by capturing graph structural evolution over time through adversarial dynamics.
We join recent work on message passing algorithms for TVGs that address the challenge of rapidly changing network topology~\cite{Altisen2021-selfstabilizingsystems}.
The local nature of our mutual exclusion problem enables us to weaken the assumptions considered by prior works in this area.
For example, we allow messages to have arbitrary but finite delays akin to asynchronous message passing, we assume weaker ``semi-synchronous'' and asynchronous models of concurrency, and we trade globally unique node identifiers for local port labels.
In Section~\ref{sec:applications}, we demonstrate how the rapid dynamics modeled by TVGs combined with these weak assumptions on node capabilities facilitate application of our mutual exclusion algorithm to both systems with \textit{passive} dynamics~\cite{Angluin2006-computationnetworks,Soloveichik2008-computationfinite}, in which nodes have no control over topological changes, and those with \textit{active} dynamics~\cite{Daymude2021-canonicalamoebot,Michail2021-distributedcomputation,Michail2017-connectivitypreserving} in which nodes control the connections they establish and sever (e.g., via movements).

\section{Preliminaries} \label{sec:prelims}

\subsection{Computational Model} \label{subsec:model}

We consider a distributed system composed of a fixed set of nodes $V$.
Each node is assumed to be \textit{anonymous}, lacking a unique identifier, and has a local memory storing its \textit{state}.
Nodes communicate with each other via message passing over a communication graph whose topology changes over time.
We model this topology using a \textit{time-varying graph} $\mathcal{G} = (V, E, T, \rho)$ where $V$ is the set of nodes, $E$ is a (static) set of undirected pairwise edges between nodes, $T = \{0, \ldots, t_{\text{max}}\}$ for some (possibly infinite) $t_{\text{max}} \in \mathbb{N}$ is called the \textit{lifetime} of $\mathcal{G}$, and $\rho : E \times T \to \{0,1\}$ is the \textit{presence} function indicating whether or not a given edge exists at a given time.
A \textit{snapshot} of $\mathcal{G}$ at time $t \in T$ is the undirected graph $G_t = (V, \{e \in E : \rho(e, t) = 1\})$ and the \textit{neighborhood} of a node $u \in V$ at time $t \in T$ is the set $N_t(u) = \{v \in V : \rho(\{u, v\}, t) = 1\}$.
For $i \geq 0$, the $i$-th \textit{round} lasts from time $i$ to the instant just before time $i+1$; thus, the communication graph in round $i$ is $G_i$.

We assume that an adversary controls the presence function $\rho$ and that $E$ is the complete set of edges on nodes $V$; i.e., we do not limit which edges the adversary can introduce.
The only constraint we place on the adversary's topological changes is $\forall t \in T, u \in V, |N_t(u)| \leq \Delta$, where $\Delta > 0$ is the fixed number of \textit{ports} per node.
When the adversary establishes a new connection between nodes $u$ and $v$, it must assign the endpoints of edge $\{u, v\}$ to open ports on $u$ and $v$ (and cannot do so if either node has no open ports).
Node $u$ locally identifies $\{u, v\}$ using its corresponding \textit{port label} $\lbl \in \{1, \ldots, \Delta\}$ and $v$ does likewise.
For convenience of notation, we use $\lbl_u(v)$ to refer to the label of the port on node $u$ that is assigned to the edge $\{u, v\}$; this mapping of port labels to nodes is not available to the nodes.
Edge endpoints remain in their assigned ports (and thus labels remain fixed) until disconnection, but nodes $u$ and $v$ may label $\{u, v\}$ differently and their labels are not known to each other a priori.
Each node has a \textit{disconnection detector} that adds the label of any port whose connection is severed to a set $\disconnectset \subseteq \{1, \ldots, \Delta\}$.
A node's disconnection detector provides it a snapshot of $\disconnectset$ whenever it starts an action execution (see below) and then resets $\disconnectset$ to $\emptyset$.\footnote{Without this assumption, the adversary could disconnect an edge assigned to port $\lbl$ of node $u$ and then immediately connect a different edge to $\lbl$, causing an indistinguishability issue for node $u$.}

Nodes communicate via message passing.
A node $u$ sends a message $m$ to a neighbor $v$ by calling $\textsc{Send}(m, \lbl_u(v))$.
Message $m$ remains in transit until either $v$ receives and processes $m$ at a later round chosen by the adversary, or $u$ and $v$ are disconnected and $m$ is lost.
Multiple messages in transit from $u$ to $v$ may be received by $v$ in a different order than they were sent.
A node always knows from which port it received a given message.

All nodes execute the same distributed algorithm $\alg$, which is a set of \textit{actions} each of the form $\langle label\rangle: \langle guard\rangle \to \langle operations\rangle$.
An action's \textit{label} specifies its name.
Its \textit{guard} is a Boolean predicate determining whether a node $u$ can execute it based on the state of $u$ and any message in transit that $u$ may receive.
An action is \textit{enabled} for a node $u$ if its guard is \true\ for $u$; a node $u$ is \textit{enabled} if it has at least one enabled action.
An action's \textit{operations} specify what a node does when executing the action, structured as
\ifarxiv

\begin{enumerate}
    \item Receiving at most one message chosen by the adversary,
    \item A finite amount of internal computation and state updates, and
    \item At most one call to $\textsc{Send}(m, \ell)$ per port label $\ell$.
\end{enumerate}
\else
(\textit{i}) receiving at most one message chosen by the adversary, (\textit{ii}) a finite amount of internal computation and state updates, and (\textit{iii}) at most one call to $\textsc{Send}(m, \ell)$ per port label $\ell$.
\fi

Each node executes its own instance of $\alg$ independently, sequentially (executing at most one action at a time), and reliably (meaning we do not consider crash or Byzantine faults).
We assume an adversary controls the timing of node activations and action executions.
When the adversary activates a node, it also chooses exactly one of the node's enabled actions for the node to execute; we note that this choice must be compatible with any message the adversary chooses to deliver to the node.
In this work, we primarily focus on \textit{semi-synchronous} activations, which we interpret in the time-varying graph context to mean that in each round, the adversary activates any (possibly empty) subset of enabled nodes concurrently and the activated nodes execute their specified actions within that round.
In Section~\ref{sec:async}, we additionally consider \textit{asynchronous} activations in which action executions may span arbitrary finite time intervals.
We only constrain the adversary by \textit{weak fairness}, meaning it must activate nodes such that any continuously enabled action is eventually executed and any message in transit on a continuously existent edge is eventually processed.

\subsection{Local Mutual Exclusion} \label{subsec:problem}

In the classical problem of \textit{mutual exclusion}, nodes enter their critical sections using locks, defined as a pair of operations \Lock\ and \Unlock\ (or ``acquire'' and ``release'').
A node issues a lock request by calling \Lock; once acquired, it is assumed that a node eventually releases these locks by calling \Unlock.
Our \textit{local mutual exclusion} variant is concerned with nodes acquiring exclusive access to themselves and their immediate neighbors, though in the present context of dynamic networks, these neighborhoods may change over time.

\ifarxiv
\begin{figure}[t]
    \centering
    \includegraphics[width=0.9\textwidth]{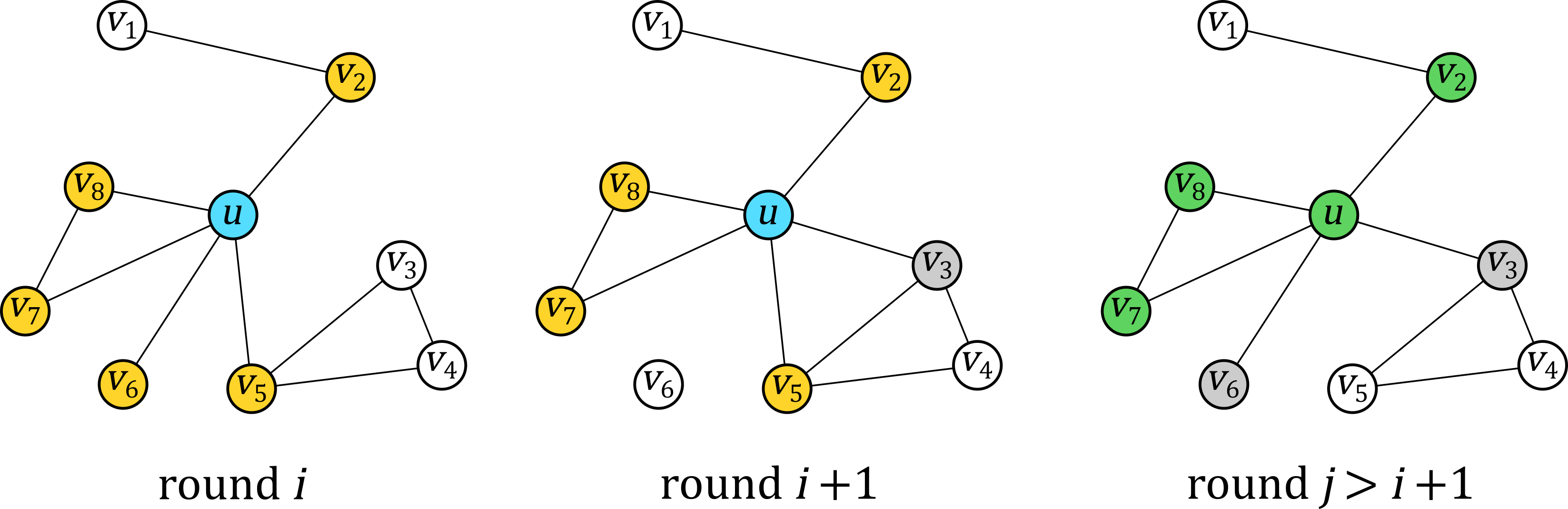}
    \caption{A node $u$ (blue) and its persistent neighborhood (yellow) over time.
    Node $u$ succeeds in locking itself and its persistent neighbors (green) in round $j$.
    Neighbors that are disconnected from $u$ before round $j$ (such as $v_6$ at time $i+1$) as well as new neighbors that connect to $u$ from time $i+1$ to $j$ (grey) are not persistent.}
    \label{fig:persistent}
\end{figure}
\else\fi

Formally, each node $u$ stores a variable $\lock \in \{\bot, 0, \ldots, \Delta\}$ that is equal to $\bot$ if $u$ is unlocked, $0$ if $u$ has locked itself, and $\lbl_u(v) \in \{1, \ldots, \Delta\}$ if $u$ is locked by $v$.
The \textit{lock set} of a node $u$ in round $i$ is $\mathcal{L}_i(u) = \{v \in N_i(u) : \lock(v) = \ell_v(u)\}$ which additionally includes $u$ itself if $\lock(u) = 0$.
Suppose that in round $i$, a node $u$ calls \Lock\ to issue a lock request of its current closed neighborhood $N_i[u] = \{u\} \cup N_i(u)$.
This lock request \textit{succeeds} at some later round $j > i$ if round $j$ is the first in which $\mathcal{L}_j(u) = \{u\} \cup \{v \in N_i(u) : \forall t \in [i, j], \{u, v\} \in G_t\}$; i.e., $j$ is the earliest round in which $u$ obtains locks for itself and every \textit{persistent} neighbor that remained connected to $u$ in rounds $i$ through
\ifarxiv
$j$ (see \figtext~\ref{fig:persistent}).
\else
$j$.
\fi
Our goal is to design an algorithm $\alg$ implementing \Lock\ and \Unlock\ that satisfies the following properties:

\begin{itemize}
    \item \textit{Mutual Exclusion.} For all rounds $i \in T$ and all pairs of nodes $u, v \in V$, $\mathcal{L}_i(u) \cap \mathcal{L}_i(v) = \emptyset$.
    
    \item \textit{Lockout Freedom.} Every issued lock request eventually succeeds with probability 1.
\end{itemize}

Following the long tradition of mutual exclusion problem definitions, our local mutual exclusion problem is defined in terms of mutual exclusion and fairness properties.
However, because each \lock\ variable points to at most one node per time, it is impossible for two nodes' lock sets to intersect, trivially satisfying the mutual exclusion property.
Nevertheless, satisfying lockout freedom remains challenging, especially in highly dynamic settings.
When issuing lock requests, nodes do not know which of their connections will remain stable and which will disconnect by the time their coordination is complete.
Thus, our problem variant captures what it means for nodes to lock their maximal persistent neighborhoods despite unpredictable and rapid topological changes.

\section{Algorithm for Local Mutual Exclusion} \label{sec:alg}

Our randomized algorithm for the local mutual exclusion problem specifies actions for the execution of \Lock\ and \Unlock\ operations satisfying mutual exclusion and lockout freedom.
An execution of the \Lock\ operation by a node $u$ is organized into two phases: a \textit{preparation phase} (Algorithm~\ref{alg:lockprep}) in which $u$ determines and notifies the nodes $L(u)$ it intends to lock, and a \textit{competition phase} (Algorithm~\ref{alg:lockcompete}) in which $u$ attempts to lock all nodes in $L(u)$, contending with any other nodes $v$ for which $L(v) \cap L(u) \neq \emptyset$.
An execution of the \Unlock\ operation (Algorithm~\ref{alg:unlock}) by node $u$ is straightforward, simply notifying all nodes in $L(u)$ that their locks are released.
All local variables used in our algorithm are listed in Table~\ref{tab:variables} as they appear in the pseudocode.
In a slight abuse of notation, we use $N[u]$ and the subsets thereof to represent both the nodes in the closed neighborhood of $u$ and the port labels of $u$ they are connected to.
For clarity of presentation, the algorithm pseudocode allows for a node to send messages to itself (via ``port 0'') just as it sends messages to its neighbors, though in reality these self-messages would be implemented with in-memory variable updates.

\begin{table}[t]
    \centering
    \caption{The notation, domain, initialization, and description of the local variables used in the algorithm for local mutual exclusion by a node $u$.}
    \label{tab:variables}
    \begin{tabular}{cp{99pt}cp{208pt}}
        \toprule
        \textbf{Var.} & \textbf{Domain} & \textbf{Init.} & \textbf{Description} \\
        \midrule
        \lock & $\{\bot, 0, \ldots, \Delta\}$ & $\bot$ & $\bot$ if $u$ is unlocked, $0$ if $u$ has locked itself, and $\lbl_u(v)$ if $u$ is locked by $v$ \\
        \state & $\{\bot, \textsc{prepare}, \textsc{compete}$, $\textsc{win}, \textsc{locked}, \textsc{unlock}\}$ & $\bot$ & The lock state of node $u$ \\
        \phase & $\{\bot, \textsc{prepare}, \textsc{compete}\}$ & $\bot$ & The algorithm phase node $u$ is in \\
        \lockset & $\subseteq N[u]$ & $\emptyset$ & Ports (nodes) $u$ intends to lock \\
        \replyset & $\subseteq N[u]$ & $\emptyset$ & Ports via which $u$ has received \msg{ready}{}, \msg{ack-lock}{}, or \msg{ack-unlock}{} responses \\
        $\winset$ & $\subseteq N[u] \times \{\true, \false\}$ & $\emptyset$ & Port-outcome pairs of \msg{win}{} messages $u$ has received \\
        \holdset & $\subseteq N[u]$ & $\emptyset$ & Ports (nodes) on hold for the competition to lock $u$ \\
        \applyset & $\subseteq N[u]$ & $\emptyset$ & Ports (nodes) of applicants that can join the competition to lock $u$ \\
        \candidateset & $\subseteq N[u]$ & $\emptyset$ & Ports (nodes) of candidates competing to lock $u$ \\
        \compete & $\subseteq C(u) \times \{0, \ldots, K-1\}$ & $\emptyset$ & Port-priority pairs of the candidates \\
        \bottomrule
    \end{tabular}
\end{table}

\begin{algorithm}[t]
\caption{The \Lock\ Operation: Preparation Phase for Node $u$} \label{alg:lockprep}
\begin{algorithmic}[1]
    \State \textsc{InitLock}: On \Lock\ being called $\to$  \Comment{Initiator initiates a lock request.}
        \Indent
            \If {$\state = \bot$} \Comment{Only one locking operation at a time.}
                \State \Call{CleanUp}{\,}.
                \State Set $\state \gets \textsc{prepare}$ and $\lockset \gets N[u]$. 
                \ForAll{$\lbl \in \lockset$} \Call{Send}{\msg{prepare}{}, $\lbl$}.
                \EndFor
            \EndIf
        \EndIndent
    \State \textsc{ReceivePrepare}: On receiving \msg{prepare}{} via port $\lbl$ $\to$
        \Indent
            \State \Call{CleanUp}{\,}.
            \If {$\phase = \textsc{compete}$} set $\holdset \gets \holdset \cup \{\lbl\}$. \Comment{Put $\lbl$ on hold if already competing.}
            \Else {}
                \State Set $\applyset \gets \applyset \cup \{\lbl\}$ and $\phase \gets \textsc{prepare}$. \Comment{Add $\lbl$ as an applicant otherwise.}
                \State \Call{Send}{\msg{ready}{}, $\lbl$}.
            \EndIf
        \EndIndent     
    \State \textsc{ReceiveReady}: On receiving \msg{ready}{} via port $\lbl$ $\to$
        \Indent
            \State \Call{CleanUp}{\,}.
            \State Set $\replyset \gets \replyset \cup \{\lbl\}$.
        \EndIndent     
    \State \textsc{CheckStart:} $(\state = \textsc{prepare}) \wedge (\replyset = \lockset)$ $\to$  \Comment{All \msg{ready}{} messages received.}
        \Indent
            \State \Call{CleanUp}{\,}.
            \State Set $\state \gets \textsc{compete}$, $\replyset \gets \emptyset$, and $\winset \gets \emptyset$.
            \State Choose priority $p \in \{0, \ldots, K-1\}$ uniformly at random.
            \ForAll{$\lbl \in \lockset$} \Call{Send}{\msg{request-lock}{$p$}, $\lbl$}.
            \EndFor
        \EndIndent
    \State \textsc{CleanUp}: $(\phase \neq \bot) \vee (\state = \textsc{unlock})$ $\to$
        \Indent
            \State \Call{CleanUp}{\,}.
        \EndIndent
    \Function {CleanUp}{\,}  \Comment{Helper function for processing disconnections $\disconnectset$.}
        \ForAll{$\lbl \in \disconnectset$}
            \If {$\lock = \lbl$} $\lock \gets \bot$.
            \EndIf
            \State Remove $\lbl$ from all sets: $\lockset \gets \lockset \setminus \{\lbl\}$, $\replyset \gets \replyset \setminus \{\lbl\}$, $\winset \gets \winset \setminus \{(\lbl, \cdot)\}$, $\holdset \gets \holdset \setminus \{\lbl\}$,
            \State $\applyset \gets \applyset \setminus \{\lbl\}$, $\candidateset \gets \candidateset \setminus \{\lbl\}$, and $\compete \gets \compete \setminus \{(\lbl, \cdot)\}$.
        \EndFor
        \If {$\candidateset = \emptyset$} 
            \ForAll {$\lbl \in \holdset$} \Call{Send}{\msg{ready}{}, $\lbl$}.
            \EndFor
            \State Set $\applyset \gets \applyset \cup \holdset$ and $\holdset \gets \emptyset$. \Comment{All nodes on hold become applicants.}
            \If {$\applyset \neq \emptyset$} set $\phase \gets \textsc{prepare}$.
            \Else {} set $\phase \gets \bot$.
            \EndIf
        \EndIf
    \EndFunction
\end{algorithmic}
\end{algorithm}

\begin{algorithm}[t]
\caption{The \Lock\ Operation: Competition Phase for Node $u$} \label{alg:lockcompete}
\begin{algorithmic}[1]
    \State \textsc{ReceiveRequest}: On receiving \msg{request-lock}{$p$} via port $\lbl$ $\to$
        \Indent
            \State \Call{CleanUp}{\,}.
            \If {$\lbl \in \applyset$} set $\applyset \gets \applyset \setminus \{\lbl\}$ and $\candidateset \gets \candidateset \cup \{\lbl\}$.
            \EndIf
            \State Set $\compete \gets \compete \cup \{(\lbl, p)\}$ and $\phase \gets \textsc{compete}$. \Comment{Close competition.}
         \EndIndent
    \State \textsc{CheckPriorities}: $(\phase = \textsc{compete}) \wedge (|\candidateset| = |\compete|)$ $\to$ \Comment{All priorities received.}
        \Indent
            \State \Call{CleanUp}{\,}.
            \If {$\lock = \bot$ and $\exists (\lbl, p) \in \compete$ with a unique highest $p$}
                \State \Call{Send}{\msg{win}{\true}, $\lbl$} and \Call{Send}{\msg{win}{\false}, $\lbl'$} for all $\lbl' \in \candidateset \setminus \{\lbl\}$.
            \Else {} \Call{Send}{\msg{win}{\false}, $\lbl$} for all $\lbl \in \candidateset$.
            \EndIf
            \State Reset $\compete \gets \emptyset$.  \Comment{Competition is over.}
        \EndIndent
    \State \textsc{ReceiveWin}: On receiving \msg{win}{$b$} via port $\lbl$ $\to$
        \Indent
            \State \Call{CleanUp}{\,}.
            \State Set $\winset \gets \winset \cup \{(\lbl, b)\}$.
        \EndIndent
     \State \textsc{CheckWin}: $(\state = \textsc{compete}) \wedge (|\winset| = |\lockset|)$ $\to$ \Comment{All \msg{win}{$b$} replies received.}
        \Indent
            \State \Call{CleanUp}{\,}.
            \If {$\exists(\cdot, \false) \in \winset$}  \Comment{Start new locking attempt.}
                \State Choose priority $p \in \{0, \ldots, K-1\}$ uniformly at random.
                \ForAll{$\lbl \in \lockset$} \Call{Send}{\msg{request-lock}{$p$}, $\lbl$}.
                \EndFor
            \Else {}  \Comment{Succeeded in locking.}
                \State Set $\state \gets \textsc{win}$ and reset $\replyset \gets \emptyset$.
                \ForAll{$\lbl \in \lockset$} \Call{Send}{\msg{set-lock}{}, $\lbl$}.
                \EndFor
            \EndIf
            \State Reset $\winset \gets \emptyset$.
        \EndIndent
    \State \textsc{ReceiveSetLock}: On receiving \msg{set-lock}{} via port $\lbl$ $\to$
        \Indent
            \State Set $\lock \gets \lbl$ and $\candidateset \gets \candidateset \setminus \{\lbl\}$.
            \State \Call{CleanUp}{\,}.
            \State \Call{Send}{\msg{ack-lock}{}, $\lbl$}.
        \EndIndent
    \State \textsc{ReceiveAckLock}: On receiving \msg{ack-lock}{} via port $\lbl$ $\to$
        \Indent
            \State \Call{CleanUp}{\,}.
            \State Set $\replyset \gets \replyset \cup \{\lbl\}$.
        \EndIndent
    \State \textsc{CheckDone}: $(\state = \textsc{win}) \wedge (\replyset = \lockset)$ $\to$  \Comment{All lock acknowledgements received.}
        \Indent
            \State \Call{CleanUp}{\,}.
            \State Set $\state \gets \textsc{locked}$ and reset $\replyset = \emptyset$.
            \State \Return $\lockset$. \Comment{Locking complete.}
        \EndIndent
\end{algorithmic}
\end{algorithm}

\begin{algorithm}[t]
\caption{The $\Unlock$ Operation for Node $u$} \label{alg:unlock}
\begin{algorithmic}[1]
    \State \textsc{InitUnlock}: On \Unlock\ being called $\to$  \Comment{Initiator initiates an unlock.}
    \Indent
        \If {$\state = \textsc{locked}$} \Comment{Only one \Unlock\ per successful \Lock.}
            \State \Call{CleanUp}{\,}.
            \State Set $\state \gets \textsc{unlock}$ and reset $\replyset \gets \emptyset$.
            \ForAll {$\lbl \in \lockset$} \Call{Send}{\msg{release-lock}{}, $\lbl$}.
            \EndFor
        \EndIf
    \EndIndent
    \State \textsc{ReceiveRelease}: On receiving \msg{release-lock}{} via port $\lbl$ $\to$
       \Indent
          \State \Call{CleanUp}{\,}.
          \State Set $\lock \gets \bot$ and \Call{Send}{\msg{ack-unlock}{}, $\lbl$}.
       \EndIndent
    \State \textsc{ReceiveAckUnlock}: On receiving \msg{ack-unlock}{} via port $\lbl$ $\to$
       \Indent
          \State \Call{CleanUp}{\,}.
          \State Set $\replyset \gets \replyset \cup \{\lbl\}$.
       \EndIndent
    \State \textsc{CheckUnlocked}: $(\state = \textsc{unlock}) \wedge (\replyset = \lockset)$ $\to$  \Comment{All unlock acknowledgements received.}
        \Indent
            \State \Call{CleanUp}{\,}.
            \State Reset $\state \gets \bot$ and $\replyset = \emptyset$. \Comment{Unlocking complete.}
        \EndIndent
\end{algorithmic}
\end{algorithm}

We refer to nodes that call \Lock/\Unlock\ as \textit{initiators} and the nodes that are being locked or unlocked as \textit{participants}; it is possible for a node to be an initiator and participant simultaneously.
Initiators progress through a series of \textit{lock states} associated with the \state\ variable; participants advance through the algorithm's \textit{phases} as indicated by the \phase\ variable.
We first describe the algorithm from an initiator's perspective and then describe the complementary participants' actions.
A special \textsc{CleanUp} helper function ensures that the nodes adapt to any disconnections affecting their variables that may have occurred since they last acted, so we omit the handling of these disconnections in the following description.

When an initiator $u$ calls \Lock, it advances to the \textsc{prepare} state, sets $\lockset(u)$ to all nodes in its closed neighborhood $N[u]$, and then sends \msg{prepare}{} messages to all nodes of $\lockset(u)$.
Once it has received \msg{ready}{} responses from all nodes of $\lockset(u)$, it advances to the \textsc{compete} state and joins the competitions for each node in $\lockset(u)$ by sending \msg{request-lock}{$p$} messages to all nodes of $\lockset(u)$, where $p$ is a priority chosen uniformly at random from $\{0, \ldots, K - 1\}$ for a fixed $K = \Theta(1)$.
It then waits for the outcomes of these competitions.
If it receives at least one \msg{win}{\false} message, it lost this competition and must compete again.
Otherwise, if all responses are \msg{win}{\true}, it advances to the \textsc{win} state and sends \msg{set-lock}{} messages to all nodes of $\lockset(u)$.
Once it has received \msg{ack-lock}{} responses from all nodes of $\lockset(u)$, it advances to the \textsc{locked} state indicating
$\lockset(u)$ now represents the lock set $\mathcal{L}(u)$.

A participant $v$ is responsible for coordinating the competition among all initiators that want to lock $v$.
To delineate successive competitions, $v$ distinguishes among initiators that are \textit{candidates} in the current competition, \textit{applicants} that may join the current competition, and those that are \textit{on hold} for the next competition.
When $v$ receives a \msg{prepare}{} message from an initiator $u$, it either puts $u$ on hold if a competition is already underway or adds $u$ as an applicant and replies \msg{ready}{} otherwise.
Participant $v$ promotes its applicants to candidates when $v$ receives their \msg{request-lock}{$p$} messages.
Once all such messages are received from the competition's candidates, $v$ notifies the one with the unique highest priority of its success and all others of their failure (or, in the case of a tie, all candidates fail).
A winning competitor is removed from the candidate set while all others remain to try again; once the candidate set is empty, $v$ promotes all initiators that were on hold to applicants.
Finally, when $v$ receives a \msg{set-lock}{} message, it sets its \lock\ variable accordingly and acknowledges this with an \msg{ack-lock}{} response.

\section{Analysis} \label{sec:analysis}

In this section, we prove the following theorem.

\begin{theorem} \label{thm:locking}
    If all nodes start with the initial values given by Table~\ref{tab:variables}, the algorithm satisfies the mutual exclusion and lockout freedom properties under semi-synchronous concurrency, requires $\bigo{\Delta}$ memory per node and messages of size $\Theta(1)$, and has at most two messages in transit along any edge at any time.
\end{theorem}

The algorithm in Section~\ref{sec:alg} is written with respect to local port labels; for ease of presentation, we use the corresponding nodes throughout this analysis and write $X_i(u)$ to denote the local variable $X$ of node $u$ at the start of round $i$.
We begin with two straightforward lemmas demonstrating the eventual execution of enabled actions.

\begin{lemma} \label{lem:enableexecute}
    Apart from \textsc{CleanUp}, every enabled action will eventually be executed.
\end{lemma}
\begin{proof}
    Any enabled \textsc{Receive*} action whose guard depends only on the receipt of some message must eventually be executed because it is assumed that every message in transit is eventually processed (unless the edge is disconnected, at which point the message is lost and the action is no longer enabled).
    Thus, it remains to consider the \textsc{Check*} actions.
    
    Suppose \textsc{CheckStart} is enabled for a node $u$ in some round $i$; i.e., $\state_i(u) = \textsc{prepare}$ and $\replyset_i(u) = \lockset_i(u)$.
    When $\state = \textsc{prepare}$, only \textsc{CheckStart} can change the $\state$ variable or reset $R$ to $\emptyset$.
    Any execution of the \textsc{CleanUp} action does not change the $\state$ variable and maintains $\replyset(u) = \lockset(u)$ since it removes any disconnected neighbors from both sets.
    So \textsc{CheckStart} remains continuously enabled and thus must eventually be executed by the weakly fair adversary.
    An analogous argument also applies to \textsc{CheckPriorities}, \textsc{CheckWin}, \textsc{CheckDone}, and \textsc{CheckUnlocked}.
\end{proof}

Lemma~\ref{lem:enableexecute} shows that an enabled action will eventually be executed, but we also need to know that the actions become enabled in the first place.
One potential obstacle is that \textsc{Check*} actions by a node $u$ need to receive all responses from the nodes in $\lockset(u)$ before becoming enabled.
If some of nodes in $\lockset(u)$ disconnect and their corresponding response messages are lost, the \textsc{Check*} action may be disabled indefinitely.
This is one role of the \textsc{CleanUp} action: removing disconnections from the algorithm's variables so other actions stop waiting for neighbors that no longer exist.
We call such an action \textit{pre-enabled} if it is currently disabled but would become enabled after \textsc{CleanUp} is executed.

\begin{lemma} \label{lem:preenabled}
    Every pre-enabled action eventually becomes enabled.
\end{lemma}
\begin{proof}
    Suppose that \textsc{CheckStart} is pre-enabled for node $u$.
    Then $\state(u) = \textsc{prepare}$ and $u$ must have sent a \msg{prepare}{} message to itself in its execution of \textsc{InitLock}.
    So \textsc{ReceivePrepare} is enabled for $u$, and by Lemma~\ref{lem:enableexecute} it is eventually executed, updating $\phase(u) = \textsc{prepare}$.
    This enables \textsc{CleanUp} for $u$, and it will remain enabled until executed because only \textsc{CleanUp} itself can reset $\phase$ to $\bot$.
    Thus, \textsc{CleanUp} must eventually be executed by the weakly fair adversary, enabling the pre-enabled \textsc{CheckStart}.

    An analogous argument also applies to \textsc{CheckPriorities}, \textsc{CheckWin}, and \textsc{CheckDone}.
    For \textsc{CheckUnlocked}, the condition $\state = \textsc{unlock}$ in the guard of \textsc{CleanUp} ensures that \textsc{CheckUnlocked} is eventually enabled.
\end{proof}

We continue our investigation of possible deadlocks resulting from actions remaining disabled by considering concurrent competitions.
An initiator node $u$ is \textit{competing} if and only if $\state(u) = \textsc{compete}$, i.e., if $u$ has executed \textsc{CheckStart} but has not yet received all \msg{win}{} messages needed to execute \textsc{CheckWin}.
We model dependencies between competing initiators and participants at the start of round $i$ as a directed bipartite graph $\mathcal{D}_i = (\mathcal{I}_i \cup \mathcal{P}_i, E_i)$ where $\mathcal{I}_i = \{u : \state_i(u) = \textsc{compete}\}$ is the set of competing initiators and $\mathcal{P}_i = \{u : \exists v \in \mathcal{I}_i \text{ s.t.\ } u \in \lockset_i(v)\}$ is the set of participants.
We note that some nodes belong to both partitions and consider their initiator and participant versions distinct.
For nodes $u \in \mathcal{I}_i$ and $v \in \mathcal{P}_i \cap \lockset_i(u)$ for which $u = v$ or $(u, v) \in G_i$ (i.e., the edge exists in round $i$), the directed edge $(u, v) \in E_i$ if and only if $u$ has not yet sent a \msg{request-lock}{} message to $v$ in response to the latest \msg{win}{} message from $v$; analogously, $(v, u) \in E_i$ if and only if $v$ has not yet sent a \msg{win}{} message to $u$ in response to the latest \msg{request-lock}{} message from $u$.

\begin{lemma} \label{lem:acyclic}
    For all rounds $i$, $\mathcal{D}_i$ is acyclic.
\end{lemma}
\begin{proof}
    Initially, no node has yet called \Lock\ and thus $\mathcal{D}_0$ is empty and trivially acyclic.
    So suppose that $\mathcal{D}_j$ remains acyclic for all rounds $0 \leq j \leq i-1$ and consider the following events that may occur in round $i-1$ to form $\mathcal{D}_i$.
    \begin{itemize}
        \item A node $u$ executes \textsc{CheckStart}.
        Then $(v, u)$ is added to $\mathcal{D}_i$ for each $v \in \lockset_{i-1}(u)$ that $u$ sends \msg{request-lock}{} messages to.
        But $u$ is a sink, so $\mathcal{D}_i$ remains acyclic.
        
        \item A node $u$ executes \textsc{CheckWin}.
        If there exists $(\cdot, \false) \in \winset_{i-1}(u)$, then $(u, v)$ is removed from $\mathcal{D}_i$ and $(v, u)$ is added to $\mathcal{D}_i$ for each $v \in \lockset_{i-1}(u)$ that $u$ once again sends \msg{request-lock}{} messages to.
        As in the first case, this makes $u$ a sink and $\mathcal{D}_i$ remains acyclic.
        Otherwise, if all $(\cdot, b) \in \winset_{i-1}(u)$ have $b = \true$, $u$ has won its competition and sets $\state_i(u) = \textsc{win}$, meaning $u \not\in \mathcal{D}_i$.
        So $\mathcal{D}_i$ remains acyclic in this case as well.
        
        \item A node $u$ executes \textsc{CheckPriorities}.
        Then $(u, v)$ is removed from $\mathcal{D}_i$ and $(v, u)$ is added to $\mathcal{D}_i$ for each $v \in \candidateset_{i-1}(u)$ that $u$ sends \msg{win}{} messages to.
        For $\mathcal{D}_i$ to be acyclic, it suffices to show it does not contain any outgoing edges from $u$; i.e., there are no nodes $w$ such that $u \in \lockset_i(w)$, $w$ has sent $u$ a \msg{request-lock}{} message, but $u$ has not yet sent a \msg{win}{} response to $w$.
        Such a node $w$ could only have sent $u$ a \msg{request-lock}{} message if it had previously received a \msg{ready}{} message from $u$, which in turn could only have been sent by $u$ if $u$ had included $w$ as an applicant in $\applyset(u)$.
        Thus, on receipt of the first \msg{request-lock}{} message from $w$, $u$ would have promoted $w$ to a candidate in $\candidateset(u)$, which is precisely the set that $u$ responds to when executing \textsc{CheckPriorities}.
        So $u$ has no outgoing edges in $\mathcal{D}_i$, as desired.
        
        \item An edge $\{u, v\}$ is disconnected in the TVG $\mathcal{G}$, for $u \in \mathcal{I}_{i-1}$ and $v \in \mathcal{P}_{i-1}$.
        This disconnection is processed by the \textsc{CleanUp} helper function, removing $v$ from $\lockset(u)$ and thus any $(u, v)$ edge from $\mathcal{D}_i$ during the next execution of \textsc{CheckWin} by $u$; an analogous statement holds for edges $(v, u)$ in the next execution of \textsc{CheckPriorities} by $v$.
        As the removal of an edge cannot create a cycle, $\mathcal{D}_i$ remains acyclic.
    \end{itemize}
    
    Therefore, $\mathcal{D}_i$ remains acyclic in all cases, as claimed.
\end{proof}

\begin{lemma} \label{lem:competedeadlock}
    Every competing initiator eventually receives a \msg{win}{} response from its participants; likewise, every participant eventually receives a \msg{request-lock}{} response from its competing initiator(s).
\end{lemma}
\begin{proof}
    Suppose to the contrary that there exists a competing initiator $u$ that waits indefinitely for a \msg{win}{} response from some participant $v$.
    Then the edge $\{u, v\}$ must never be disconnected in the TVG $\mathcal{G}$ and the directed edge $(v, u)$ must remain indefinitely in $\mathcal{D}$.
    By Lemmas~\ref{lem:enableexecute} and~\ref{lem:preenabled}, $v$ can only be prohibited from sending the requisite \msg{win}{} message if \textsc{CheckPriorities} remains disabled for $v$ indefinitely.
    This, in turn, is only possible if $v$ waits indefinitely for a \msg{request-lock}{} response from some competing initiator $w \neq u$.
    This implies that $\{v, w\}$ is never disconnected in $\mathcal{G}$ and the directed edge $(w, v)$ remains indefinitely in $\mathcal{D}$.
    As before, Lemmas~\ref{lem:enableexecute} and~\ref{lem:preenabled} can be applied iteratively to show that each node must be waiting on another.
    But since the set of nodes $V$ is finite, some node must eventually be revisited, establishing a directed cycle in $\mathcal{D}$ and contradicting Lemma~\ref{lem:acyclic}.
\end{proof}

Lemma~\ref{lem:competedeadlock} directly implies the following corollary.

\begin{corollary} \label{cor:finitetrials}
    Every competition trial of a competing initiator eventually completes.
\end{corollary}

To demonstrate that our algorithm satisfies lockout freedom, it remains to show that every competing initiator $u$ eventually wins a competition trial by receiving all \msg{win}{\true} responses from $\lockset(u)$.
We first address the situation in which a competition trial of $u$ is \textit{open}, meaning none of the nodes $v \in \lockset(u)$ are locked during the trial.

\begin{lemma} \label{lem:opentrials}
    If $K = \Theta(1)$, then an initiator that competes in an open competition trial infinitely often will eventually win a competition, with probability 1.
\end{lemma}
\begin{proof}
    Consider any competing initiator $u$ and any open competition trial of $u$.
    By the start of its second competition trial, $u \in \candidateset(v)$ for all $v \in \lockset(u)$, implying that $\phase(v) = \textsc{compete}$ and no other nodes will be added to $\candidateset(v) \cup \applyset(v)$ while $u$ is still competing for $v$.
    Since $|\lockset(u) \setminus \{u\}| \leq \Delta$ and $|\candidateset(v) \cup \applyset(v) \setminus \{u\}| \leq \Delta$ for each $v \in \lockset(u) \setminus \{u\}$, node $u$ can be competing against $c \leq \Delta^2$ other nodes.
    Every node chooses its priority uniformly at random from $\{0, \ldots, K - 1\}$, so it follows from symmetry that the probability $u$ has the highest priority in a given trial is at least $1/\Delta^2$.
    In general,
    \begin{align*}
        \hspace{-5mm}\frac{\Pr{p(u) \text{ highest} \mid p(u) \text{ unique}}}{\Pr{p(u) \text{ highest}}} &= \frac{\sum_{p=0}^{K-1} \Pr{p(u) = p \, \wedge \, \forall v \neq u : p(v) \leq p(u) \mid p(u) \text{ unique}}}{\sum_{p=0}^{K-1} \Pr{p(u) = p \, \wedge \, \forall v \neq u : p(v) \leq p(u)}} \\
        &= \frac{\sum_{p=0}^{K-1} \frac{1}{K}\left(\frac{p}{K-1}\right)^c}{\sum_{p=0}^{K-1} \frac{1}{K}\left(\frac{p+1}{K}\right)^c}
        \geq \frac{\sum_{p=0}^{K-1} p^c}{\sum_{p=1}^K p^c}
        \geq \frac{(K - 1)^c}{2K^c}
        \geq \frac{(1 - 1/K)^{\Delta^2}}{2}
    \end{align*}
    Furthermore, the probability that $u$ has a unique priority is $(1 - 1/K)^c \geq (1 - 1/K)^{\Delta^2}$.
    Thus, the probability that $u$ has the unique highest priority in a given open trial is
    \begin{align*}
        \Pr{p(u) \text{ highest} \wedge p(u) \text{ unique}} &= \Pr{p(u) \text{ highest} \mid p(u) \text{ unique}} \cdot \Pr{p(u) \text{ unique}} \\
        &\geq \frac{(1 - 1/K)^{\Delta^2}}{2\Delta^2} \cdot (1 - 1/K)^{\Delta^2} = \frac{(1 - 1/K)^{2\Delta^2}}{2\Delta^2} > 0.
    \end{align*}
    Since this probability is strictly positive, the probability that $u$ never has the unique highest priority in an infinite sequence of open competition trials is
    \[\lim_{n \to \infty}(1 - \Pr{p(u) \text{ highest} \wedge p(u) \text{ unique}})^n
    \leq \lim_{n \to \infty} \left(1 - \frac{(1 - 1/K)^{2\Delta^2}}{2\Delta^2}\right)^n = 0.\]
    Therefore, with probability 1 there must eventually be an open competition trial in which $u$ has the unique highest priority.
    Because the trial is open, all $v \in L(u)$ have $\lock(v) = \bot$ and thus will send \msg{win}{\true} responses to $u$.\footnote{This proof can be easily extended to show that if $K > \Delta^2$, $u$ will win a competition within $\bigo{\Delta^2}$ open competition trials, in expectation.
    We chose to avoid this increase in message size requirements from $\Theta(1)$ to $\bigo{\log\Delta}$ since time complexity is not a focus of this work.}
\end{proof}

We next show that a competing initiator competes in an open trial infinitely often.
Recall from Section~\ref{subsec:problem} that a \Lock\ operation by node $u$ succeeds once $u$ obtains locks for its persistent neighborhood, and once obtained, these locks are eventually released via \Unlock.

\begin{lemma} \label{lem:success}
    Every competing initiator eventually wins a competition trial with probability~1.
\end{lemma}
\begin{proof}
    Suppose to the contrary that a competing initiator $u$ competes in an infinite number of competition trials.
    Only a finite number of these trials can be open, since $u$ would eventually win one of an infinite number of open trials with probability 1 by Lemma~\ref{lem:opentrials}.
    So an infinite number of trials of $u$ must be closed; i.e., there are an infinite number of trials in which at least one $v \in \lockset(u)$ has $\lock(v) \neq \bot$.
    Since $|\lockset(u) \setminus \{u\}| \leq \Delta$, there must be a node $v \in \lockset(u)$ that is locked infinitely often.
    But by the start of its second competition trial, $u \in \candidateset(v)$ and no other nodes will be added to $\candidateset(v) \cup \applyset(v)$ while $u$ is still competing for $v$.
    Thus, only the nodes in $\candidateset(v) \cup \applyset(v)$ and the node that had already locked $v$ when $u$ was added to $\candidateset(v)$ could possibly lock $v$.
    But whenever $v$ sets its locks in \textsc{ReceiveSetLock}, it removes the locking node from $\candidateset(v)$.
    Moreover, any node that obtains locks must eventually release them, by supposition.
    So the set of nodes that could lock $v$ is monotonically decreasing and thus nodes in $\candidateset(v) \cup \applyset(v) \setminus \{u\}$ cannot lock $v$ an infinite number of times, a contradiction.
\end{proof}

For an initiator $u$ to benefit from eventual victory ensured by Lemma~\ref{lem:success}, it must become competing in the first place; i.e., it must advance to $\state(u) = \textsc{compete}$.

\begin{lemma} \label{lem:prepare}
    Every initiator eventually becomes competing.
\end{lemma}
\begin{proof}
    Suppose to the contrary that an initiator $u$ never becomes competing, i.e., it never executes \textsc{CheckStart}.
    By Lemmas~\ref{lem:enableexecute} and~\ref{lem:preenabled}, this is only possible if \textsc{CheckStart} remains disabled indefinitely.
    To be an initiator at all, $u$ must have executed \textsc{InitLock}, set $\state(u) = \textsc{prepare}$, and sent \msg{prepare}{} messages to all nodes $v \in \lockset(u)$.
    So $u$ must be waiting for a \msg{ready}{} response from at least one $v \in \lockset(u)$ that remains connected to $u$ indefinitely.
    
    By Lemma~\ref{lem:enableexecute}, such a node $v$ must eventually execute \textsc{ReceivePrepare}.
    During this execution, it must be the case that $\phase(v) = \textsc{compete}$ and $v$ adds $u$ to $\holdset(v)$; otherwise, $v$ would have added $u$ to $\applyset(v)$ and replied to $u$ with a \msg{ready}{} message, a contradiction.
    Only the \textsc{CleanUp} helper function can reset $\phase(v)$ to $\bot$, but it only does so when $\candidateset(v) \cup \applyset(v) \cup \holdset(v) = \emptyset$ which is not the case since $u \in \holdset(v)$.
    So the \textsc{CleanUp} action is continuously enabled for $v$ and is eventually executed by the weakly fair adversary.
    During this execution, it must be the case that $\candidateset(v) \neq \emptyset$; otherwise, $v$ would have sent \msg{ready}{} messages to all initiators on hold at $v$, including $u$, a contradiction.
    But for this situation to occur indefinitely, there must exist some competitor in the finite set $\candidateset(v) \cup \applyset(v)$ that competes in an infinite number of trials, a contradiction of Lemma~\ref{lem:success}.
\end{proof}

Combining Corollary~\ref{cor:finitetrials} with Lemmas~\ref{lem:success} and~\ref{lem:prepare} implies the following corollary.

\begin{corollary}
    The local mutual exclusion algorithm satisfies lockout freedom.
\end{corollary}

Recall from Section~\ref{subsec:problem} that the mutual exclusion property is trivially satisfied by our construction of the lock sets.
Thus, we conclude the proof of Theorem~\ref{thm:locking} with the following result regarding the algorithm's memory and message size requirements.

\begin{lemma}
    The algorithm requires $\bigo{\Delta}$ memory per node and messages of size $\Theta(1)$, and there are at most two messages in transit along any given edge at any time.
\end{lemma}
\begin{proof}
    Table~\ref{tab:variables} shows that \phase\ and \state\ can be stored in $\Theta(1)$ bits each and \lock\ can be stored in $\log_2\Delta$ bits.
    The remaining variables can be represented as linear registers of length $\Delta$, where port $\ell$ is in the set variable $X$ if and only if the $\ell$-th bit of register $X$ is \true.
    So the memory bound of $\bigo{\Delta}$ follows.
    Similarly, there are a constant number of message types, among which only \msg{request-lock}{} and \msg{win}{} carry additional data.
    A \msg{win}{} message carries one bit signaling whether a competition trial was won or lost.
    A \msg{request-lock}{} message carries a randomly chosen priority, which by Lemma~\ref{lem:opentrials} can be stored in $\Theta(1)$ bits.
    
    To bound the number of messages in transit per edge per time, consider the execution of a \Lock\ operation by an initiator node $u$.
    The local mutual exclusion algorithm is structured around pairs of initiator messages and participant responses: \msg{prepare}{}/\msg{ready}{} messages in the preparation phase, \msg{request-lock}{}/\msg{win}{} messages in each competition trial, and \msg{set-lock}{}/\msg{ack-lock}{} messages once a node has won a trial.
    In each scenario, only one message per pair is in transit along $\{u, v\}$ per time for each $v \in \lockset(u)$.
    Moreover, node $u$ does not advance to the next phase and send any additional messages until all messages of the current phase are processed.
    An analogous argument applies to the \Unlock\ operation with its \msg{release-lock}{}/\msg{ack-unlock}{} message pairs.
    Thus, there can be at most one message in transit per edge per time involved with any initiator's \Lock\ or \Unlock\ operation.
    
    Furthermore, an initiator $u$ can execute at most one \Lock\ or \Unlock\ operation per time since $u$ can only start a \Lock\ operation by executing \textsc{InitLock} if $\state(u) = \bot$, implying it holds no locks; similarly, $u$ can only start an \Unlock\ operation by executing \textsc{InitUnlock} if $\state(u) = \textsc{locked}$, implying its previous \Lock\ operation has succeeded.
    
    Thus, the lemma follows since there are at most two initiators $u$ and $v$ per edge $\{u, v\}$.
\end{proof}

\section{Extending to Asynchronous Concurrency} \label{sec:async}

Section~\ref{sec:analysis} proved Theorem~\ref{thm:locking} under semi-synchronous concurrency in which (\textit{i}) topological changes occur at discrete times in between rounds of action executions and (\textit{ii}) the adversary chooses any non-empty subset of nodes to act in each round and those nodes' action executions are guaranteed to end before the next round begins.
In this section, we prove that Theorem~\ref{thm:locking} holds even in the more general asynchronous setting.

All assumptions from Section~\ref{subsec:model} about the time-varying graph $\mathcal{G}$, the nodes, their asynchronous message passing, and the structure of algorithms and their actions remain the same.
However, in an \textit{asynchronous schedule}, the adversary can schedule action executions over arbitrary finite time intervals, including those that are concurrent with topological changes and span multiple TVG rounds.
In this setting, our prior assumptions about the disconnection detector now imply that any topological changes incident to $u$ that are concurrent with one of its action executions are not observed or processed by $u$ until its next action execution.
We further assume for the asynchronous setting that any message sent by node $u$ during one of its action executions starting at time $t_1$ is processed by a node $v$ during some other action execution starting at time $t_2 > t_1$ if and only if the edge $\{u, v\} \in G_t$ for all $t \in [t_1, t_2]$.
This implies that when an edge is disconnected, all messages in transit along that edge are immediately lost and no further messages can be sent or received by the corresponding ports until the corresponding action executions have finished.

\begin{lemma} \label{lem:reduction}
    For any asynchronous schedule $\mathcal{S}$, there exists a semi-synchronous schedule $\mathcal{S}'$ containing the same action executions as in $\mathcal{S}$ that produces the same outcome for every action execution in $\mathcal{S}$.
\end{lemma}
\begin{proof}
    Consider any asynchronous schedule $\mathcal{S}$ of the local mutual exclusion algorithm and let $\mathcal{E}$ be the set of all action executions in $\mathcal{S}$.
    Analogous to Lamport~\cite{Lamport1978-timeclocks}, we define the \textit{causal relation} $\to$ on $\mathcal{E}$ as the smallest relation satisfying the following three conditions:
    \begin{itemize}
        \item If $\alpha \in \mathcal{E}$ is an execution by node $u$ and $\beta \in \mathcal{E}$ is the next execution by $u$, then $\alpha \to \beta$.
        \item If a message sent in $\alpha \in \mathcal{E}$ is processed in $\beta \in \mathcal{E}$, then $\alpha \to \beta$.
        \item If $\alpha \to \beta$ and $\beta \to \gamma$, then $\alpha \to \gamma$.
    \end{itemize}
    Since all causal relations are naturally forward in time, it follows that the graph represented by the causal relations on $\mathcal{E}$ forms a DAG.
    Thus, the action executions of $\mathcal{E}$ can be topologically sorted in some order $[\alpha_1, \alpha_2, \ldots]$.
    
    Now, consider the schedule $\hat{\mathcal{S}}$ containing the same action executions starting at the same times as those in $\mathcal{S}$, but (\textit{i}) each action execution takes 0 time and (\textit{ii}) any set of action executions starting at the same time as some edge changes is shifted before these edge changes without changing the order of the action executions.
    Then $\hat{\mathcal{S}}$ can be transformed into a semi-synchronous schedule $\mathcal{S}'$ by adding filler time steps when no edges change so that each node executes at most one action per round and all action executions between two time steps start at the same time.
    Certainly, $\mathcal{S}'$ is still a valid schedule since all causal relations remain forward in time and---by our assumption on asynchronous message processing---any message sent by action execution $\alpha$ that is processed by action execution $\beta$ in $\mathcal{S}$ can still be processed by $\beta$ in $\mathcal{S}'$.
    Furthermore, since the causal relations haven't changed, the action executions in $\mathcal{S}'$ can be sorted in the same order $[\alpha_1, \alpha_2, \ldots]$ as for $\mathcal{S}$.
    Since any action execution can only change a node's state or send messages and, in both schedules, it only sees a snapshot of $\disconnectset$ at its start, it follows by induction on the ordering of the action executions that for any $i$, the outcome of $\alpha_i$ is identical in $\mathcal{S}$ and $\mathcal{S}'$.
\end{proof}

As in the semi-synchronous setting, the mutual exclusion property is trivially satisfied in the asynchronous setting.
But suppose to the contrary that there exists an asynchronous schedule in which at least one \Lock\ operation never succeeds.
Lemma~\ref{lem:reduction} shows that there must exist a semi-synchronous schedule in which at least one \Lock\ operation never succeeds, contradicting Theorem~\ref{thm:locking}.
So we have the following corollary.

\begin{corollary}
    The local mutual exclusion algorithm also satisfies the mutual exclusion and lockout freedom properties under any asynchronous schedule.
\end{corollary}

\section{Applications} \label{sec:applications}

We next establish how our algorithm for local mutual exclusion can be used to implement key assumptions present in formal models of dynamic distributed systems.
In particular, we focus on the assumptions of independent pairwise interactions in \textit{population protocols}~\cite{Angluin2006-computationnetworks} and concurrency control operations in the \textit{canonical amoebot model} of programmable matter~\cite{Daymude2021-canonicalamoebot}.

\subparagraph*{Population Protocols.}

Inspired by passively mobile sensor networks, Angluin et al.~\cite{Angluin2006-computationnetworks} proposed the \textit{population protocols} model.
Each of agent in a population is assumed to have a finite state and a transition function defining how that state evolves as a result of a pairwise interaction with another agent.
Agents cannot explicitly control their movements or who they interact with; i.e., they are passively dynamic.
Instead, it is typically assumed that a sequential scheduler chooses one pair of agents to interact per time step.
In reality, however, many agents within interacting distance might exist concurrently (see, e.g.,~\cite{Czumaj2021-trulyparallel}), requiring a mechanism to organize these agents into a matching of independent pairs.

This goal could be achieved directly using our algorithm for local mutual exclusion.
Any agent $u$ that wants to interact must first call \Lock.
On success, $u$ then chooses any locked neighbor to interact with, if it has one; if desired, one could even generalize the usual pairwise interactions to interactions among the full group of locked neighbors.
Lockout freedom ensures $u$ will eventually be allowed to make this choice, and mutual exclusion ensures this pairwise interaction is isolated from any others.
After interacting, $u$ then releases its locks with \Unlock.
If the expected number of competing agents is high, an alternative implementation of our algorithm could have $u$ make its choice of interacting neighbor $v$ first and then try to lock only $u$ and $v$ to avoid a lengthy competition.
On success, $u$ would then interact with $v$, isolated from any other interactions, and then unlock itself and $v$.
In both implementations, it is possible that all neighbors may move out of interaction range, leaving $u$ to lock only itself.
In this situation, no interaction occurs and $u$ simply unlocks itself.

Both implementations require $\bigo{\Delta}$ memory per agent and messages of size $\Theta(1)$.
For many applications of population protocols where $\Delta$ is a fixed constant (e.g., proximity graphs or IoT), these requirements are reduced to $\Theta(1)$.
Thus, our algorithm for local mutual exclusion could provide isolated pairwise interactions assumed by population protocols even in the presence of underlying network dynamics and asynchronous concurrency.

\subparagraph*{The Canonical Amoebot Model.}

The \textit{amoebot model} abstracts active programmable matter as a collection of simple computational elements called \textit{amoebots} that move and interact locally to collectively achieve tasks of coordination and movement.
Each amoebot is typically assumed to be anonymous and have only constant-size memory, but can control its movements.
The \textit{canonical amoebot model}~\cite{Daymude2021-canonicalamoebot} is an updated formalization that addresses concurrency by partitioning amoebot functionality into a high-level application layer where algorithms call various operations and a low-level system layer where those operations are executed via asynchronous message passing.
Two such operations are the concurrency control operations, \Lock\ and \Unlock, which are used in a \textit{concurrency control framework} to convert amoebot algorithms that terminate in the sequential setting and satisfy certain conventions into algorithms that exhibit equivalent behavior in the concurrent setting~\cite{Daymude2021-canonicalamoebot}.

The canonical amoebot model treats the \Lock\ and \Unlock\ operations as black boxes without giving an implementation.
These operations facilitate amoebots gaining exclusive access to themselves and their neighbors, much like our mutual exclusion property, and are assumed to terminate (either successfully or in failure) in finite time.
The asynchronous extension of our local mutual exclusion algorithm presented in Section~\ref{sec:async} could directly implement these operations, ensuring isolation of concurrent amoebot actions even as connections between amoebots change due to their movements.
One interesting feature of such an implementation is that while the amoebot \Lock\ operation is allowed to fail---which must be taken into account by algorithm designers---our \Lock\ operation always succeeds due to lockout freedom, reducing complexity in algorithm design.
Moreover, for the often-considered geometric space variant in which an (expanded) amoebot can have at most eight neighbors, our algorithm has $\Theta(1)$ amoebot memory and message size requirements.

\section{Conclusion} \label{sec:conclude}

We presented an algorithm for local mutual exclusion that enables weakly capable nodes to isolate concurrent actions involving their persistent neighborhoods despite dynamic network topology.
Our algorithm ensures that nodes belong to at most one locked neighborhood at a time (mutual exclusion) and that every lock request eventually succeeds (lockout freedom).
It requires $\bigo{\Delta}$ memory per node and messages of size $\Theta(1)$---where $\Delta$ is the maximum number of connections per node---and is compatible with anonymous, message passing nodes that operate semi-synchronously or asynchronously.
These weak requirements make our algorithm suitable for a wide range of application domains such as overlay networks, IoT, modular robots, and programmable matter.
As two concrete examples, we demonstrated how our algorithm could implement the pairwise interactions assumed by population protocols~\cite{Angluin2006-computationnetworks} and the concurrency control operations assumed by the canonical amoebot model~\cite{Daymude2021-canonicalamoebot}.

\bibliographystyle{plainurl}
\bibliography{ref}

\end{document}